\documentclass[journal]{article}
\usepackage[letterpaper, total={7in, 8in}]{geometry}
\usepackage{setspace}

\usepackage{times}
\usepackage{mathtools}

\usepackage{amssymb, tikz}
\usepackage{amsthm}
\usepackage{amsxtra}
\usepackage{amsmath}
\usepackage{array}
\usepackage{algpseudocode}
\usepackage{algorithm}

\usepackage{verbatim}
\usepackage{epsfig}
\usepackage{setspace}
\usepackage{caption}
\usepackage{subcaption}
\usepackage{breqn}
\usepackage{color}
\usepackage{enumerate}
\usepackage{eufrak}
\usepackage{paralist}

\graphicspath{Figures/}

\allowdisplaybreaks

\newtheorem{theorem}{Theorem}

\newtheorem{lemma}[theorem]{Lemma}

\newtheorem*{remark*}{Remarks}

\newcommand{\cA}{\mathcal{A}}
\newcommand{\cG}{\mathcal{G}}
\newcommand{\cV}{\mathcal{V}}
\newcommand{\cE}{\mathcal{E}}
\newcommand{\cM}{\mathcal{M}}
\newcommand{\field}{\mathbb{F}}
\newcommand{\code}{\mathcal{C}}
\newcommand{\rate}{\mathcal{R}}
\newcommand{\rank}{\text{rank}}
\newcommand{\cN}{\mathcal{N}}

\title{Optimal storage codes on graphs with fixed locality}
\author{{Sabyasachi Basu}\dag \hspace{1cm}\and{Manuj Mukherjee}\ddag}
\date{}

\begin{document}
\maketitle

\renewcommand{\thefootnote}{}
\footnotetext{
 \noindent\dag Sabyasachi Basu is with the University of California Santa Cruz (UCSC), USA. Email:sbasu3@ucsc.edu. 
 
 \ddag Manuj Mukherjee is with the Indraprastha Institute of Information Technology Delhi (IIIT Delhi), India. Email: manuj@iiitd.ac.in.}

\renewcommand{\thefootnote}{\arabic{footnote}}

\begin{abstract}
    Storage codes on graphs are an instance of \emph{codes with locality}, which are used in distributed storage schemes to provide local repairability. Specifically, the nodes of the graph correspond to storage servers, and the neighbourhood of each server constitute the set of servers it can query to repair its stored data in the event of a failure. A storage code on a graph with $n$-vertices is a set of $n$-length codewords over $\field_q$ where the $i$th codeword symbol is stored in server $i$, and it can be recovered by querying the neighbours of server $i$ according to the underlying graph. 

    In this work, we look at binary storage codes whose repair function is the parity check, and characterise the tradeoff between the locality of the code and its rate. Specifically, we show that the maximum rate of a code on $n$ vertices with locality $r$ is bounded between $1-1/n\lceil n/(r+1)\rceil$ and $1-1/n\lceil n/(r+1)\rceil$. The lower bound on the rate is derived by constructing an explicit family of graphs with locality $r$, while the upper bound is obtained via a lower bound on the binary-field rank of a class of symmetric binary matrices. Our upper bound on maximal rate of a storage code matches the upper bound on the larger class of codes with locality derived by Tamo and Barg. As a corollary to our result, we obtain the following asymptotic separation result: given a sequence $r(n), n\geq 1$, there exists a sequence of graphs on $n$-vertices with storage codes of rate $1-o(1)$ if and only if $r(n)=\omega(1)$. 
    \end{abstract}

\section{Introduction}\label{sec:intro}

In the era of `big data' and cloud storage large files are required to be fragmented into smaller components and stored across multiple servers. The simplest storage scheme consists of partitioning the file and storing each fragment in a different server. Unfortunately, this scheme is not robust to server-failure. How does one recover a part of the file if the server storing it fails? A naive solution is simply replicating the different fragments and storing the replicas across multiple servers, thereby ensuring that every part of the file is recoverable even if one of the servers fail. This replication-based strategy however wastes storage resources. Instead, an efficient approach to robust distributed storage of large files involves encoding it with a \emph{code with locality} \cite{Gopalan,BargTamo}.

Let $\field_q$ denote the Galois field of size $q$. A code with locality $r$ is an error control code $\code\subseteq\field_q^n$, such that for every $i\in[n]$, there exists at most $r$ locations $j_1,j_2,\ldots,j_r\in[n]$, such that for any codeword $c\triangleq(c_1,c_2,\ldots,c_n)\in\code$, its $i$th codeword symbol $c_i$ can be recovered from the codeword symbols $c_{j_k}, 1\leq k\leq r$. Therefore, for robust storage, a large file can be encoded using a code with locality, with each codeword symbol being stored in different servers. In the event of a server failure, the failed server can recover its content by querying at most $r$ servers. A code with locality $r$, however, only guarantees that to repair a failed server, at most $r$ servers need to be queried. In some situations, say due to the geographical location of servers, a fixed server might not be able to query any arbitrary set of $r$ servers. To tackle this issue, Mazumdar \cite{Arya1} and Shanmugam and Dimakis \cite{SD14} simultaneously introduced the concept of \emph{storage codes on graphs}.

A storage code on a graph is defined as follows. Consider a graph $\cG=(\cV,\cE)$ on $n$ vertices, i.e., $|\cV|\;=n$. For any vertex $v\in\cV$, denote the neighbours of $v$ by $N(v)\triangleq\{u\in\cV:(u,v)\in\cE\}$. The vertices of the graph represent the servers, and the neighbourhood of a vertex represents the set of servers it can query from. We restrict ourselves to graphs without isolated vertices,\footnote{We call a vertex isolated if no edge connects to it.} since every server must be able to query at least one server to repair its content in case of a failure. Each vertex $v\in\cV$ is assigned a symbol from $\field_q$, and is equipped with a repair function $f_v:\field_q^{|N(v)|}\to\field_q$, which is used to repair its stored symbol by querying its neighbours' symbols. A storage code $\code\subseteq\field_q^n$ on the graph $\cG$ consists of the set of $n$-length strings $c=(c_1,c_2,\ldots,c_n)$ from $\field_q$, such that for any $v\in\cV$, we have $f_v(c_j:j\in N(v))=c_v$. 

The locality of a storage code on a graph is therefore the \emph{maximum degree} of the graph. The \emph{rate} of a storage code $\code\subseteq\field_q$ on a graph $\cG$ is defined as $\log_q|\code|/n$. In particular, if the code is \emph{linear},\footnote{A code $\code\subseteq\field_q^n$ is linear if it is a subspace of $\field_q^n$.} the rate of the code is given by $\dim(\code)/n$. The \emph{capacity} of a graph $\cG$ is defined to be the maximal rate among all possible storage codes on $\cG$. Note that the definition of rate implies a trivial bound of one on the capacity.

In their pioneering work \cite{Arya1}, Mazumdar was interested in characterising the capacity of an arbitrary graph $\cG$. The capacity was shown to be upper bounded by $n$ minus the \emph{minrank} \cite{Haemers} of the graph, while lower bounds on the capacity were given through specific constructions. In a follow-up work, Mazumdar et al. \cite{MMV19} obtained tighter upper bounds on the capacity via \emph{gadget covering}. The authors further showed that their upper bound was tight for certain classes of graphs (see Section~V of \cite{MMV19}).

A different thread in the study of storage codes on graphs was initiated by Barg and Zemor in \cite{BZ22}. It is a well-known fact that complete graphs admit storage codes of rate 1, i.e., the maximal possible rate. This is achieved using codes over the binary field $\field_2$ with the repair function being a parity check, i.e., $f_v(c_j:j\in N(v))=\sum_{j\in N(v)}c_j$, for all $v\in\cV$, where the summation is modulo-2. The authors in \cite{BZ22} wanted to design a sequence of \emph{`sparse'} graphs allowing binary storage codes of rate asymptotically approaching one as the number of vertices increase to infinity. In \cite{BZ22}, the authors were able to come up with a sequence of \emph{triangle-free} graphs which allow binary storage codes of rates approaching $3/4$ asymptotically. In a subsequent work \cite{BSY22}, the problem was completely solved by constructing a sequence of triangle-free graphs allowing binary storage codes with rates approaching one asymptotically.

\subsection{Our contribution}\label{sec:contrib}

In this work, we continue along the lines of recent work~\cite{BZ22,BSY22} and ask the following question: Given a sequence of localities $r(n), n\geq 1$, when is there a sequence of graphs on $n$ vertices allowing binary storage codes of rate $1-o(1)$ with locality $r(n)$? Our main result is the following theorem, which exactly characterises what sequences of localities allow graphs with storage codes approaching rate one.

\begin{theorem}
    Given a sequence of localities $r(n), n\geq 1$, there exists a sequence of graphs $\cG_n, n\geq 1$, on $n$-vertices supporting binary storage codes of locality $r(n)$ and rate $1-o(1)$ if and only if $r(n)=\omega(1)$. 
    \label{th:main}
\end{theorem}

To prove Theorem~\ref{th:main}, we first construct a sequence of graphs with locality $r(n)$ and rate $1-\frac{1}{n}\lceil\frac{n}{r(n)+1}\rceil$. Our constructed graph consists of disjoint cliques of size at most $r(n)+1$. Since the constructed code has rate $1-\frac{1}{n}\lceil\frac{n}{r(n)+1}\rceil$, $r(n)=\omega(1)$ immediately implies a rate of $1-o(1)$. We also modify our original construction to come up with a sequence of connected graphs with locality $r(n)$ supporting binary storage codes of rate $1-\frac{3}{n}\lceil\frac{n}{r(n)+1}\rceil$. Hence, even these graphs have storage codes with rates $1-o(1)$ whenever $r(n)=\omega(1)$. 

To prove the converse, i.e., no binary storage code with locality $r(n)=O(1)$ and rate $1-o(1)$ exist, we first follow \cite{BSY22} to relate the rate of the binary storage code on the graph $\cG$ to the $\field_2$-rank of its \emph{augmented adjacency matrix}. Note that if $\cG$ has locality $r(n)$, then every row of its augmented adjacency matrix can have at most $r(n)$ ones. We then prove a lower bound on the $\field_2$-ranks of binary matrices with rows having at most $r(n)$ ones, which gives the converse. In fact, the lower bound actually shows that for any given $n\geq 1$ and any locality $r(n)$, our construction of disjoint cliques is optimal in terms of rate.  

We remark that the converse of Theorem~\ref{th:main} also follows from the more general result of Tamo and Barg \cite[Theorem~2.1]{BargTamo} on codes with locality. In particular, since Tamo and Barg's result holds for any code with locality, their result shows that the converse part Theorem~\ref{th:main} remain valid for codes on any field $\field_q$ with any repair function $f_v$. However, while Tamo and Barg's proof relies on a modification of the Turan's theorem for graphs, our proof is built upon a lower bound on the $\field_2$-rank of symmetric binary matrices with rows of bounded Hamming weight, which may be of independent interest.

\subsection{Related work}\label{sec:rel}

\subsubsection{Index codes}\label{sec:indcode}

 A closely related setup to the problem of storage codes on graphs is the problem of \emph{index coding with side information} \cite{index}. The index coding problem considers a (possibly directed) graph $\cG=(\cV,\cE)$, with $|\cV|=n$, where each vertex $v$ wishes to compute a symbol $c_v\in\field_q$. The vertices $v\in\cV$ have access to $(c_{v'}:v'\in N(v))$. The goal of index coding is to design a function $f:\field_q^n\to\field_q^k$ with the least possible $k$, as well as local decoding functions $g_v:\field_q^{|N(v)|}\times\field_q^k\to\field_q$ at each vertex $v\in\cV$, satisfying $c_v=g_v\biggl(f(c_1,\ldots,c_n),(c_{v'}:v'\in N(v))\biggr)$. It is immediate to see that the \emph{nullspace} of a linear index code corresponds to a storage code on the graph with $g_v$ being the respective repair functions. Furthermore, cosets of any linear storage code on a graph gives rise to a linear index coding scheme. These observation were summarised in \cite[Proposition~1]{Arya1}, and this connection was further exploited to obtain bounds on the storage capacity of graphs in \cite{Arya1,MMV19}. A detailed exposition on the relationship between index coding and storage codes on graphs is available in \cite{Arya3}.

 \subsubsection{Hat-guessing games on graphs}\label{sec:hatguess}

 The problem of \emph{hat guessing games} on graphs was introduced in \cite{Riis}, and later expanded upon in subsequent works like \cite{undirectedHGN,CDRguessing}. In a hat guessing game with $n$ players, each player sits on the nodes of a graph $\cG=(\cV,\cE)$ with $|\cV|=n$. The players are provided with hats with colours from $\field_q$, which they are unable to see. However, the player in node $v\in\cV$ can see the colours of the players in $N(v)$. Before the assignment of the hats, the players fix a guessing strategy, i.e., the players agree upon functions $f_v:\field_q^{|N(v)|}\to\field_q$, for all $v\in \cV$, using which they are going to guess the colour of their hats based on the colours of their neighbours' hats. The hat guessing game then tries to come up with the guessing strategy that maximizes the probability of all parties guessing correctly when the colours of the hats are assigned uniformly at random. 

 It is not difficult to see that if the guessing functions $f_v$ correspond exactly to those of a storage code $\code\subseteq\field_q^n$ on $\cG$, then exactly the hat assignments from $\code$ can be correctly guessed at all parties. Thus, the probability of correct guessing in case of a uniformly at random assignment of hat colours equates to the rate of the storage code corresponding to the guessing strategy. Therefore, the guessing strategy that maximizes the probability of correct guessing corresponds to the maximal rate storage code on $\cG$. These connections between storage codes, hat guessing games on graphs, and index coding with side information are available in detail in \cite{FK}.

 \subsubsection{Codes with `local' properties}\label{sec:local}

 As remarked earlier, the storage codes on graphs are a special case of codes with locality \cite{Gopalan}, which attempt to minimize the number of servers accessed in order to repair the information stored at a server. A related notion is that of \emph{regenerating codes}, which attempt at minimizing the amount of information exchanged in order to repair the information at a server \cite{Dim,PVKetal,NRK}. Other `local' properties are also demanded from codes used for distributed storage, such as \emph{local decodability} \cite{TK,Yekhanin}, and \emph{compression with local decode and update} \cite{Shashank1,Shashank2}. 

 \subsection{Notations and organisation}\label{sec:org}

 For any $n\in\mathbb{N}$, we shall use the notation $[n]\triangleq\{1,2,\ldots,n\}$. We use $\rank(A)$ to denote the $\field_2$-rank of any square binary matrix $A$. Moreover, we shall use the word `rank' to mean $\field_2$-rank throughout this paper.
 
 Basic definitions and the main result that proves Theorem~\ref{th:main} is introduced in Section~\ref{sec:prelims}. The achievability proof of Theorem~\ref{th:main} via construction of rate-optimal storage codes for fixed locality is detailed in Section~\ref{sec:achievable}. The converse to Theorem~\ref{th:main} is proved in Section~\ref{sec:conv}. The paper concludes with Section~\ref{sec:conc}.

\section{Preliminaries}\label{sec:prelims}

Recall that a binary storage code $\code\subseteq \{0,1\}^n$ with parity repair functions on the graph $\cG=(\cV,\cE)$, with $|\cV|=n$, is the set of vectors of the form $c=(c_v)_{v\in\cV}$ satisfying $c_v=\sum_{u\in N(v)}c_u$\footnote{The sum is over the binary field.} for every $v\in \cV$. The \emph{locality} $r$ of the graph $\cG$ is defined as the maximum number of bits that need to be queried in case a node fails. The locality is therefore given by the maximum-degree $\max_{v\in \cV}|N(v)|$. The efficiency of the storage code is measured in terms of the \emph{rate} of the code defined as $\rate(\code)=\log|\code|/n$.

Note that a graph $\cG$ on vertices has locality $r(n)$ if and only if all the rows of its \emph{augmented adjacency matrix}\footnote{The augmented adjacency matrix of a graph is simply its adjacency matrix, but with its diagonal elements being 1.} $\bar{A}(\cG)$ have weights $w(i)\triangleq |\{j:\bar{A}(i,j)=1\}|$ satisfying $2\leq w(i)\leq r+1$.\footnote{We need $w(i)\geq 2$ since every node, in the event of a failure, must be able to query at least one more node to repair its content.} The following elementary lemma relates the rate of the code $\code$ on $\cG$ with the rank of the augmented adjacency matrix $\bar{A}(\cG)$.

\begin{lemma}
\label{lem:basic}
A binary storage code $\code_\cG$ on a graph $\cG=(\cV,\cE)$ over $n$ vertices is a linear code with rate 
$$
\rate(\code_\cG)=1-\frac{\rank(\bar{A}(\cG))}{n}.
$$
\end{lemma}
\begin{proof}
Since the repair functions are parity checks, the code is linear and its parity-check matrix corresponds to $\bar{A}(\cG)$. The result then follows via the rank-nullity theorem. See Section~1 of \cite{BSY22} for details.
\end{proof}

Lemma~\ref{lem:basic} tells us that among all graphs on $n$ vertices, the graph with the best possible rate is the complete graph $K_n$, yielding a rate of $1-(\rank(\bar{A}(K_n)))/n=1-1/n$. In fact, it is easy to see that the storage code corresponding to $K_n$ is the so-called \emph{single parity check} code. The issue with $K_n$ however is that each vertex has degree $n-1$, meaning that recovering a bit involves querying all the remaining bits. 

In this paper, we attempt to identify the tradeoff between locality and rate. More precisely, define the capacity of binary storage codes on graphs of length $n$ and locality $r$ to be
\begin{equation}
    C_{n,r}\triangleq \sup \rate(\code_\cG), \label{eq:capacity}
\end{equation}
where the supremum is over all possible graphs on $n$ vertices and locality at most $r$. In this paper, we shall prove the following theorem that characterises $C_{n,r}$.

\begin{theorem}\label{th:cap}
    For any $n\geq 2$ and any $2\leq r\leq n-1$, the capacity of binary storage codes on graphs of length $n$ and locality $r$ is given by 
    $$
    1-\frac{1}{n}\biggl\lceil\frac{n}{r+1}\biggr\rceil \leq C_{n,r}\leq 1-\frac{1}{n}\biggl\lfloor\frac{n}{r+1}\biggr\rfloor.
    $$
\end{theorem}

Observe that the lower and upper bounds on $C_{n,r}$ almost match -- they are off by at most $1/n$, a factor that disappears asymptotically. 

Theorem~\ref{th:cap} is proved in the two subsequent sections. We conclude this section by noting that Theorem~\ref{th:cap} implies Theorem~\ref{th:main}. 

\section{Proof of the achievable part of Theorem~\ref{th:cap}}\label{sec:achievable}

In this section, we prove that $C_{n,r}\geq 1-\frac{1}{n}\lceil\frac{n}{r+1}\rceil$. To do so, in the following lemma, we explicitly construct a graph $\cG$ on $n$ vertices with locality $r$, satisfying $\rate(\code_\cG)=1-\frac{1}{n}\lceil\frac{n}{r+1}\rceil$.

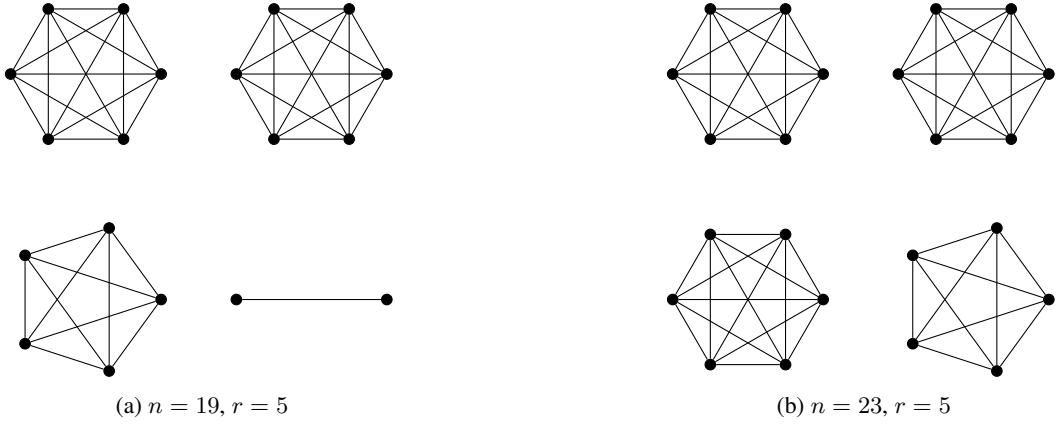
\begin{figure}
    \centering
    \begin{subfigure}[b]{0.49\textwidth}
    \centering
        \begin{tikzpicture}[scale=1]
    
    % Node properties
    \tikzset{vertex/.style={draw, circle, fill=black, inner sep=0pt, minimum size=4pt}}
    
    % First part
    \begin{scope}[shift={(0,0)}]
        \foreach \i in {1,...,6} {
            \node[vertex] (A\i) at ({\i*60+60}:1) {};
        }
        \foreach \i in {1,...,6} {
        \foreach \j in {\i,...,6} {
            \ifnum\j>\i
                \draw (A\i) -- (A\j);
            \fi
            }
        }
    \end{scope}
    
    % Second part
    \begin{scope}[shift={(3,0)}]
        \foreach \i in {1,...,6} {
            \node[vertex] (B\i) at ({\i*60+60}:1) {};
        }
        \foreach \i in {1,...,6} {
        \foreach \j in {\i,...,6} {
            \ifnum\j>\i
                \draw (B\i) -- (B\j);
            \fi
            }
        }
    \end{scope}
    
    % Third part
    \begin{scope}[shift={(0,-3)}]
        \foreach \i in {1,...,5} {
            \node[vertex] (D\i) at ({\i*72+72}:1) {};
        }
        \foreach \i in {1,...,5} {
        \foreach \j in {\i,...,5} {
            \ifnum\j>\i
                \draw (D\i) -- (D\j);
            \fi
            }
        }
    \end{scope}
    
    % Fourth part
    \begin{scope}[shift={(3,-3)}]
        \foreach \i in {1,...,2} {
            \node[vertex] (D\i) at ({\i*180+180}:1) {};
        }
        \foreach \i in {1,...,2} {
        \foreach \j in {\i,...,2} {
            \ifnum\j>\i
                \draw (D\i) -- (D\j);
            \fi
            }
        }
    \end{scope}
    
    \end{tikzpicture}
    \label{fig:n19r5}
    \caption{$n=19$, $r=5$}
    \end{subfigure}
    \begin{subfigure}[b]{0.49\textwidth}
    \centering
        \begin{tikzpicture}[scale=1]
    
    % Node properties
    \tikzset{vertex/.style={draw, circle, fill=black, inner sep=0pt, minimum size=4pt}}
    
    % First part
    \begin{scope}[shift={(0,0)}]
        \foreach \i in {1,...,6} {
            \node[vertex] (A\i) at ({\i*60+60}:1) {};
        }
        \foreach \i in {1,...,6} {
        \foreach \j in {\i,...,6} {
            \ifnum\j>\i
                \draw (A\i) -- (A\j);
            \fi
            }
        }
    \end{scope}
    
    % Second part
    \begin{scope}[shift={(3,0)}]
        \foreach \i in {1,...,6} {
            \node[vertex] (B\i) at ({\i*60+60}:1) {};
        }
        \foreach \i in {1,...,6} {
        \foreach \j in {\i,...,6} {
            \ifnum\j>\i
                \draw (B\i) -- (B\j);
            \fi
            }
        }
    \end{scope}
    
    % Third part
    \begin{scope}[shift={(0,-3)}]
        \foreach \i in {1,...,6} {
            \node[vertex] (C\i) at ({\i*60+60}:1) {};
        }
        \foreach \i in {1,...,6} {
        \foreach \j in {\i,...,6} {
            \ifnum\j>\i
                \draw (C\i) -- (C\j);
            \fi
            }
        }
    \end{scope}
    
    % Fourth part
    \begin{scope}[shift={(3,-3)}]
        \foreach \i in {1,...,5} {
            \node[vertex] (D\i) at ({\i*72+72}:1) {};
        }
        \foreach \i in {1,...,5} {
        \foreach \j in {\i,...,5} {
            \ifnum\j>\i
                \draw (D\i) -- (D\j);
            \fi
            }
        }
    \end{scope}
    
    \end{tikzpicture}
    \label{fig:n23r5}
    \caption{$n=23$, $r=5$}
    \end{subfigure}
    \caption{Examples of graphs from Lemma~\ref{lem:const} for both the cases, i.e., when $n\mod r+1 = 1$ $(n=19, r=6)$, and when $n \mod r+1\neq 1$ $(n=23, r=5)$.}
    \label{fig:example-family}
\end{figure}

\begin{lemma}
\label{lem:const} Given any $n\geq 2$ and any $2\leq r\leq n-1$, there exists a graph $\cG$ on $n$ vertices with locality $r$ with
$$
\rate(\code_\cG) = 1-\frac{1}{n}\biggl\lceil\frac{n}{r+1}\biggr\rceil.
$$
\end{lemma}

\begin{proof}
    For brevity, let us define $p\triangleq\lceil\frac{n}{r+1}\rceil$. We shall construct a graph $\cG$ on $n$ vertices, whose augmented adjacency matrix $\bar{A}(\cG)$ has rows with weights between 2 and $r+1$, and satisfies $\rank(\bar{A}(\cG))=p$. Then, by Lemma~\ref{lem:basic}, the storage code $\code_\cG$ will satisfy $\rate(\code_\cG)=1-\frac{1}{n}\lceil\frac{n}{r+1}\rceil$ as required. 

    We construct the graph $\cG$ as follows. Identify $\cV$ as $[n]$, and define a partition of $\cV$ as $\{\cV_1,\cV_2,\ldots, \cV_p\}$ as follows. If $n\mod (r+1) \neq 1$, then define $\cV_i=\{(i-1)(r+1)+1,\ldots, i(r+1)\}, 1\leq i\leq p-1$, and $\cV_p=\{(p-1)(r+1)+1,\ldots,n\}$. On the other hand, if $n\mod (r+1)=1$, then set $\cV_i=\{(i-1)(r+1)+1,\ldots, i(r+1)\}, 1\leq i\leq p-2$, $\cV_{p-1}=\{(p-2)(r+1)+1,\ldots,(p-1)(r+1)-1\}$, and $\cV_p=\{(p-1)(k+1),\ldots,n\}$. The graph $\cG$ is then constructed by assigning edges to pairs of vertices $v_1,v_2\in[n]$ if and only $v_1,v_2\in\cV_i$ for some $1\leq i\leq p$. Examples of the constructed graphs are shown in Figure~\ref{fig:example-family}

    First, observe that the $(u,v)$-th entry of the augmented adjacency matrix $\bar{A}(\cG)_{u,v}=1$ if and only if $u=v$ or $u,v\in\cV_i$ for some $i$. Thus, if $v\in\cV_i$, the weight of the corresponding row in $\bar{A}(\cG)$ is $|\cV_i|$. By construction, $2\leq |\cV_i|\leq r+1$, for all $1\leq i\leq p$, and hence $\cG$ has locality $r$. Next, select fixed candidates $v_i$ from each part $\cV_i,1\leq i\leq p$, and observe that the set of rows corresponding to the $v_i$s are linearly independent. Furthermore, any $v\notin\{v_1,v_2,\ldots,v_p\}$ must belong to some $\cV_j$ for some $1\leq j\leq p$. Then, the rows corresponding $v$ and $v_j$ are exactly the same. Hence, the rows corresponding to the set $\{v_1,v_2,\ldots,v_p\}$ form a basis of the rowspace of $\bar{A}(\cG)$. Thus, $\rank(\bar{A}(\cG))=p$ as required.
\end{proof}

To conclude this section, we note that while Lemma~\ref{lem:const} constructed disconnected graphs achieving the capacity lower bound, we can tweak the construction slightly to obtain connected graphs with locality $r$ achieving a rate of at least $1-\frac{3}{n}\lceil\frac{n}{r+1}\rceil$. Note that these graphs are enough to prove the achievability part of Theorem~\ref{th:main}, but we do not show that they are optimal for a fixed $n$ and $r$. The exact construction is given in the lemma that follows.

\begin{lemma}
    \label{lem:connected}
    For any $n\geq 2$ and any $2\leq r\leq n-1$, there exists a connected graph $\cG=(\cV,\cE)$ on $n$ vertices with locality $r$, whose associated binary storage code satisfies $\rate(\code_\cG)\geq 1-\frac{3}{n}\biggl\lceil\frac{n}{r+1}\biggr\rceil$.
\end{lemma}

\begin{proof}
    For brevity, we define $p\triangleq \lceil\frac{n}{r+1}\rceil$, and identify $\cV$ with $[n]$. We begin by defining the required connected graph as follows. Firstly, partition the set of vertices $\cV$ into the sets $\cV_i, 1\leq i\leq p$, where $\cV_i=\{(i-1)(r+1)+1,(i-1)(r+1)+2,\ldots,i(r+1)\}$, for $1\leq i\leq p-1$, and $\cV_{p}=\{(p-1)(r+1)+1,\ldots,n\}$. Thus, by construction $|\cV_i|\;=r+1, 1\leq i\leq p-1$, and $|\cV_{p}|\;=n-(r+1)(p-1)\leq r+1$.

    Next, we define the set of edges $\cE$ as follows. Fix any $i$, where $1\leq i\leq p-1$, and connect every pair of vertices in $\cV_i$ with an edge, except the pair $((i-1)(r+1)+1,i(r+1))$. Next, if $|\cV_{p}|\;\geq 3$, then connect every pair of vertices in $\cV_{p}$ by an edge, except the pair $((p-1)(r+1)+1,n)$. On the other hand, if $\cV_{p}=\{(p-1)(r+1)+1, n\}$, then connect the vertices $(p-1)(r+1)$ and $n$ by an edge. Finally, connect the pairs $(i(r+1),i(r+1)+1), 1\leq i\leq p-1$. Examples of the graph $\cG$ with $k=r$ and $N=9,10,11,$ and $12$ appear in Figure~\ref{fig:conngraphs}.

    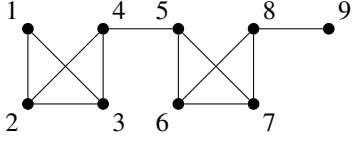
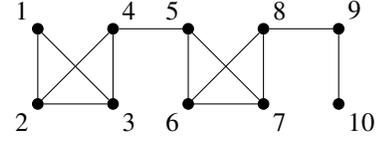
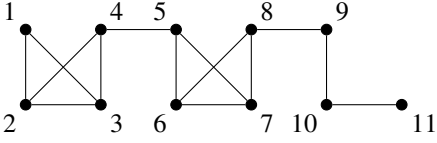
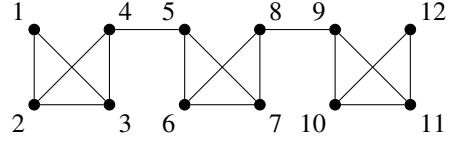
\begin{figure*}
        \centering
        \begin{subfigure}{0.45\textwidth}
            \centering
            \begin{tikzpicture}
            \tikzset{vertex/.style={draw, circle, fill=black, inner sep=0pt, minimum size=4pt}}

                \node[vertex] (A1) at (0,2) {};
                \node[vertex] (A2) at (1,3) {};
                \node[vertex] (A3) at (1,2) {};
                \node[vertex] (A4) at (0,3) {};
                \node[vertex] (A5) at (2,3) {};
                \node[vertex] (A6) at (3,3) {};
                \node[vertex] (A7) at (3,2) {};
                \node[vertex] (A8) at (2,2) {};
                \node[vertex] (A9) at (4,3) {};

                \draw(A1)--(A2);
                \draw(A1)--(A3);
                \draw(A1)--(A4);
                \draw(A3)--(A2);
                \draw(A4)--(A3);
                \draw(A2)--(A5);
                \draw(A7)--(A8);
                \draw(A5)--(A7);
                \draw(A5)--(A8);
                \draw(A7)--(A6);
                \draw(A8)--(A6);
                \draw(A6)--(A9);

                \node[below left] at (A1) {2};
                \node[above right] at (A2) {4};
                \node[below right] at (A3) {3};
                \node[above left] at (A4) {1};
                \node[below left] at (A8) {6};
                \node[above right] at (A6) {8};
                \node[below right] at (A7) {7};
                \node[above left] at (A5) {5};
                \node[above right] at (A9) {9};
            \end{tikzpicture}
            \caption{$n=9, r=3$, $\cV_1=\{1,2,3,4\}$, $\cV_2=\{5,6,7,8\}$, $\cV_3=\{9\}$}
            \label{fig:n9k3}
        \end{subfigure}
        \hfill
        \begin{subfigure}[b]{0.45\textwidth}  
            \centering
            \begin{tikzpicture}
            \tikzset{vertex/.style={draw, circle, fill=black, inner sep=0pt, minimum size=4pt}}

                \node[vertex] (A1) at (0,2) {};
                \node[vertex] (A2) at (1,3) {};
                \node[vertex] (A3) at (1,2) {};
                \node[vertex] (A4) at (0,3) {};
                \node[vertex] (A5) at (2,3) {};
                \node[vertex] (A6) at (3,3) {};
                \node[vertex] (A7) at (3,2) {};
                \node[vertex] (A8) at (2,2) {};
                \node[vertex] (A9) at (4,3) {};
                \node[vertex] (A10) at (4,2) {};

                \draw(A1)--(A2);
                \draw(A1)--(A3);
                \draw(A1)--(A4);
                \draw(A3)--(A2);
                \draw(A4)--(A3);
                \draw(A2)--(A5);
                \draw(A7)--(A8);
                \draw(A5)--(A7);
                \draw(A5)--(A8);
                \draw(A7)--(A6);
                \draw(A8)--(A6);
                \draw(A6)--(A9);
                \draw(A9)--(A10);

                \node[below left] at (A1) {2};
                \node[above right] at (A2) {4};
                \node[below right] at (A3) {3};
                \node[above left] at (A4) {1};
                \node[below left] at (A8) {6};
                \node[above right] at (A6) {8};
                \node[below right] at (A7) {7};
                \node[above left] at (A5) {5};
                \node[above right] at (A9) {9};
                \node[below right] at (A10) {10};
            \end{tikzpicture}
            \caption{$n=10, r=3$, $\cV_1=\{1,2,3,4\}$, $\cV_2=\{5,6,7,8\}$, $\cV_3=\{9,10\}$}    
            \label{fig:n10k3}
        \end{subfigure}
        \vskip\baselineskip
        \begin{subfigure}[b]{0.45\textwidth}   
            \centering 
            \begin{tikzpicture}
            \tikzset{vertex/.style={draw, circle, fill=black, inner sep=0pt, minimum size=4pt}}

                \node[vertex] (A1) at (0,2) {};
                \node[vertex] (A2) at (1,3) {};
                \node[vertex] (A3) at (1,2) {};
                \node[vertex] (A4) at (0,3) {};
                \node[vertex] (A5) at (2,3) {};
                \node[vertex] (A6) at (3,3) {};
                \node[vertex] (A7) at (3,2) {};
                \node[vertex] (A8) at (2,2) {};
                \node[vertex] (A9) at (4,3) {};
                \node[vertex] (A10) at (4,2) {};
                \node[vertex] (A11) at (5,2) {};
                
                \draw(A1)--(A2);
                \draw(A1)--(A3);
                \draw(A1)--(A4);
                \draw(A3)--(A2);
                \draw(A4)--(A3);
                \draw(A2)--(A5);
                \draw(A7)--(A8);
                \draw(A5)--(A7);
                \draw(A5)--(A8);
                \draw(A7)--(A6);
                \draw(A8)--(A6);
                \draw(A6)--(A9);
                \draw(A9)--(A10);
                \draw(A10)--(A11);
                %\draw(A9)--(A11);

                \node[below left] at (A1) {2};
                \node[above right] at (A2) {4};
                \node[below right] at (A3) {3};
                \node[above left] at (A4) {1};
                \node[below left] at (A8) {6};
                \node[above right] at (A6) {8};
                \node[below right] at (A7) {7};
                \node[above left] at (A5) {5};
                \node[above right] at (A9) {9};
                \node[below left] at (A10) {10};
                \node[below right] at (A11) {11};
            \end{tikzpicture}
            \caption{$n=11, r=3$, $\cV_1=\{1,2,3,4\}, \cV_2=\{5,6,7,8\}, \cV_3=\{9,10,11\}$}    
            \label{fig:n11k3}
        \end{subfigure}
        \hfill
        \begin{subfigure}[b]{0.45\textwidth}   
            \centering 
            \begin{tikzpicture}
            \tikzset{vertex/.style={draw, circle, fill=black, inner sep=0pt, minimum size=4pt}}

                \node[vertex] (A1) at (0,2) {};
                \node[vertex] (A2) at (1,3) {};
                \node[vertex] (A3) at (1,2) {};
                \node[vertex] (A4) at (0,3) {};
                \node[vertex] (A5) at (2,3) {};
                \node[vertex] (A6) at (3,3) {};
                \node[vertex] (A7) at (3,2) {};
                \node[vertex] (A8) at (2,2) {};
                \node[vertex] (A9) at (4,3) {};
                \node[vertex] (A10) at (4,2) {};
                \node[vertex] (A11) at (5,2) {};
                \node[vertex] (A12) at (5,3) {};
                
                \draw(A1)--(A2);
                \draw(A1)--(A3);
                \draw(A1)--(A4);
                \draw(A3)--(A2);
                \draw(A4)--(A3);
                \draw(A2)--(A5);
                \draw(A7)--(A8);
                \draw(A5)--(A7);
                \draw(A5)--(A8);
                \draw(A7)--(A6);
                \draw(A8)--(A6);
                \draw(A6)--(A9);
                \draw(A9)--(A10);
                \draw(A10)--(A11);
                \draw(A11)--(A12);
                \draw(A10)--(A12);
                \draw(A9)--(A11);

                \node[below left] at (A1) {2};
                \node[above right] at (A2) {4};
                \node[below right] at (A3) {3};
                \node[above left] at (A4) {1};
                \node[below left] at (A8) {6};
                \node[above right] at (A6) {8};
                \node[below right] at (A7) {7};
                \node[above left] at (A5) {5};
                \node[above left] at (A9) {9};
                \node[below left] at (A10) {10};
                \node[below right] at (A11) {11};
                \node[above right] at (A12) {12};
            \end{tikzpicture}
            \caption{$n=12, r=3$, $\cV_1=\{1,2,3,4\}, \cV_2=\{5,6,7,8\}, \cV_3=\{9,10,11,12\}$}    
            \label{fig:n12k3}
        \end{subfigure}
        \caption{Examples of graphs $\cG$ appearing in the proof of Lemma~\ref{lem:connected} for $r=3$ and $n=9,10,11,$ and $12$.} 
        \label{fig:conngraphs}
    \end{figure*}

    We shall first argue that $\cG$ has locality $r$, i.e., the degree of each vertex of $\cG$ is at most $r$. Firstly, consider the vertices of $\cV_i, 1\leq i\leq p-1$. Vertices $(i-1)(r+1)+j, 2\leq j\leq r$ have degree $r$ since it has all the remaining $r$ vertices of $\cV_i$. Also, the vertices $i(r+1), 1\leq i\leq p-1$ have degree $r$ since they have as their neighbours $i(r+1)+1$ as well as every vertex in $\cV_i$ except $(i-1)(r+1)+1$. Finally, vertex 1 has degree $r-1$, being connected to vertices $2,3,\ldots,r$. Now, for the vertices in $\cV_p$, if $|\cV_p|\geq 3$, the maximum degree of any vertex is at most $|\cV_p|-1\leq r$. Else, the maximum degree of vertices in $\cV_p$ is $2\leq r$. Therefore, all vertices have degree at most $r$ as required. To complete the proof, we use Lemma~\ref{lem:basic}, and show that $\bar{A}(\cG)$ has rank at most $3p$.

    Let $\bar{A}(\cG)_i$ denote the $i$th row of the augmented adjacency matrix of $\cG$. Let $I\subsetneq [n]$ denote the set of vertices defined as follows. For every $1\leq i\leq p-1$, the vertices $(i-1)(r+1)+1, (i-1)(r+1)+2, i(r+1)$ are contained in $I$. Next, if $(r+1)(p-1)+1=n$, assign $(r+1)(p-1)+1$ to $I$. Otherwise, if $(r+1)(p-1)+2=n$, assign both $(r+1)(p-1)+1$ and $(r+1)(p-1)+2$ to $I$. Else, if $(r+1)(p-1)+2<n$, assign $(r+1)(p-1)+1, (r+1)(p-1)+2$, and $n$ to $I$. Therefore, by construction, $|I|\;\leq 3p$. 

    We shall argue that either $I=[n]$, or for every $i\notin I$, there exists an $i'\in I$ such that $\bar{A}(\cG)_i=\bar{A}(\cG)_{i'}$. This would ensure that $\rank(\bar{A}(\cG))\leq |I|\;\leq 3p$, which would complete the proof. To see this, first observe that if $r=2$, then $I=[n]$. Otherwise, if $r\geq 3$, note that $[n]\setminus I=\cup_{1\leq j\leq p}I_j$, where $I_j\triangleq \{(j-1)(r+1)+l:3\leq l\leq r\}$, for $1\leq j\leq p-1$, and $I_p$ is given as follows. If $n\leq (r+1)(p-1)+3$, then $I_p=\emptyset$. Else, when $n\geq (r+1)(p-1)+4$, then $I_p=\{(p-1)(r+1)+l:3\leq l\leq n-(p-1)(r+1)-1\}$. Observe that $I_j\subsetneq \cV_j, 1\leq j\leq p$.
    
    Now, fix any $j$ with $1\leq j\leq p-1$. Then, by construction, any vertex $v\in I_j$ has exactly the same set of neighbours\footnote{Since we are considering the augmented adjacency matrix, by default we consider that every vertex is a neighbour of itself.} as the vertex $(j-1)(r+1)+2\in I$. Thus, $\bar{A}(\cG)_v=\bar{A}(\cG)_{(j-1)(r+1)+2}$. Finally, if $I_p\neq \emptyset$, then any $v\in I_p$ has exactly the same neighbours as $(p-1)(r+1)+2\in I$, and hence $\bar{A}(\cG)_v=\bar{A}(\cG)_{(p-1)(r+1)+2}$. Hence, for every $i\notin I$, we have shown that there exists $i'\in I$ satisfying $\bar{A}(\cG)_i=\bar{A}(\cG)_{i'}$, and so the lemma follows.
\end{proof}

\section{Proof of the converse part of Theorem~\ref{th:cap}}\label{sec:conv}

We conclude the proof of Theorem~\ref{th:cap} in this section by showing that for any $n\geq 2$ and $r\leq n-1$, any symmetric binary matrix $A\in\{0,1\}^{n\times n}$ with 1's in the diagonal and whose rows have a weight of at most $r+1$, has a rank of at least $\lfloor\frac{n}{r+1}\rfloor$. Along with Lemma~\ref{lem:basic}, this would then imply that $C_{n,r}\leq 1-\frac{1}{n}\lfloor\frac{n}{r+1}\rfloor$.

\begin{lemma}
    \label{lem:rank} Let $\cA\in\{0,1\}^{n\times n}$ be the augmented adjacency matrix of a graph on $n$ vertices such that no vertex has more than $r$ neighbors. Then $$\rank(\cA) \geq \biggl\lfloor\frac{n}{r+1}\biggr\rfloor.$$
\end{lemma}
\begin{proof}
    \iffalse The idea is to add unique rows to the set $N$ such that no row is a linear combination of any other two rows. \fi Let the set of columns of $\cA$ be $C =\{c_1,\ldots,c_n\}$. Each column (respectively row) is a $n\times 1$ vector in $\field_2^n.$ Denote the $i$-th entry of the column $c_j$ by $c_j(i)$ .
    
    We begin by building a set $\cN$ of `nice' columns of $\cA$ as follows. We iterate over $i\in[n]$. For each $i$, if there is no column $c_j\in \cN$ such that the $c_j(i)=1$, then add a new column $c$ from $C\setminus \cN$ such that  $c(i) = 1$. If there is no such column, we skip to the next $i$, and we terminate once we reach $n$. It follows therefore that the set of columns in $\cN$ are linearly independent. It is therefore enough to show that $|\cN|\geq \lfloor\frac{n}{r+1}\rfloor$.
    
    To bound $|\cN|$, we begin by giving an alternate interpretation of the procedure for constructing $\cN$ as follows. We identify the set of vertices with $[n]$, and pick columns corresponding to certain vertices and add them to $\cN$: in ascending order, for every vertex $i\in[n]$, we add the column $c_j$ for some $j\in N(i)\cup\{i\}$ iff no other $c_k$, with $k\in N(i)\cup\{i\}$, has already been added to $\cN$. Since, each vertex has at most $r$ neighbors, selecting a column for any $i$ ensures that for at most $|N(i)\cup\{i\}|\;\leq r+1$ of the remaining vertices, no columns can be added to $\cN$. Hence, at least $\lfloor\frac{n}{r+1}\rfloor$ `nice' columns are added to $\cN$. 

    Therefore, $|\cN|\geq \lfloor\frac{n}{r+1}\rfloor$, and since it is a linearly independent set by construction, the lower bound on the rank follows.
\end{proof}

\section{Conclusion}\label{sec:conc}

In this work, we considered binary storage codes with parity repair functions, and showed that for any given sequence of localities $r(n), n\geq 1$, we have a sequence of graphs $\cG_n$ with rates $1-o(1)$ if and only if $r(n)=\omega(1)$. This result was shown by characterising the capacity of binary storage codes on graphs with length $n$ and locality $r$. The capacity upper bound was obtained via a lower bound on the ranks of symmetric binary matrices with bounded row-weights, while the capacity lower bound was achieved via an explicit construction of a family of graphs. The obtained upper and lower bounds on the capacity differ by at most $1/n$, a factor which disappears asymptotically.

A relevant future direction of research is characterising the tradeoff between the number of edges of a graph without isolated vertices and the rate of the binary storage code supported on it. It already follows from our construction that as long as $|\cE|=\omega(n)$, a rate of $1-o(1)$ is possible. To see this, simply use the graphs from Lemma~\ref{lem:const} with $r=\Theta(|\cE|/n)=\omega(1)$. On the other hand, if $|\cE|=o(n)$, then we must have isolated vertices, and hence such graphs cannot support storage codes. The case when $|\cE|=\Theta(n)$ remains to be explored.

\bibliographystyle{acm}
\bibliography{ref}
\end{document}